\documentclass[hyperref,final]{llncs}

\usepackage[utf8]{inputenc}
\usepackage[T1]{fontenc}
\usepackage{lmodern}
\usepackage[english]{babel}
\usepackage{a4wide}

\usepackage{amsfonts,amssymb,amsmath}

\usepackage{xspace}
\usepackage{xparse}
\usepackage{xcolor,colortbl}
\usepackage{hhline}
\usepackage[basic]{complexity}
\renewclass\P{PTIME}
\renewclass\EXP{EXPTIME}
\renewclass\coNEXP{coNEXPTIME}
\usepackage{tikz}
\usetikzlibrary{arrows,automata,shapes,calc,positioning,intersections,backgrounds}

\tikzstyle{player1}=[draw,rounded rectangle, minimum size=7mm]
\tikzstyle{player2}=[draw,rectangle,minimum size=7mm]
\tikzstyle{widget}=[draw,ellipse,dashed,minimum size=6mm]
\tikzset{every loop/.style={looseness=7},every
  node/.style={thick}, every edge/.style={thick, draw}}

\usepackage{paralist}
\usepackage[ruled,boxed,vlined,linesnumbered]{algorithm2e}

\usepackage{url}
\usepackage{subfigure}

\newcommand\N{\ensuremath{\mathbb{N}}\xspace}
\renewcommand\R{\ensuremath{\mathbb{R}}\xspace}
\newcommand\Rplus{\ensuremath{\R_{\geq 0}}\xspace}
\newcommand\Rbar{\ensuremath{\R_\infty}\xspace}
\newcommand\Z{\ensuremath{\mathbb{Z}}\xspace}

\renewcommand\leq{\leqslant}
\renewcommand\geq{\geqslant}

\newcommand\powerset[1]{\ensuremath{2^{#1}}\xspace}

\newcommand\MinPl{\ensuremath{\mathsf{Min}}\xspace}
\newcommand\MaxPl{\ensuremath{\mathsf{Max}}\xspace}

\newcommand\game{\ensuremath{\mathcal G}\xspace}

\newcommand\Vertices{\ensuremath{V}\xspace}
\newcommand\VerticesMin{\ensuremath{\Vertices_{\MinPl}}\xspace}
\newcommand\VerticesMax{\ensuremath{\Vertices_{\MaxPl}}\xspace}
\newcommand\VerticesT{\ensuremath{\Vertices_t}\xspace}
\newcommand\vertex{\ensuremath{v}\xspace}

\newcommand\alphabet{\ensuremath{A}\xspace}

\newcommand\Edges{\ensuremath{E}\xspace}
\newcommand\edge{\ensuremath{e}\xspace}

\newcommand\weight{\ensuremath{\mathsf{Weight}}\xspace}

\newcommand\FinitePlays{\ensuremath{\mathsf{Plays}}\xspace}
\newcommand\FinitePlaysMin{\FinitePlays^\MinPl}
\newcommand\FinitePlaysMax{\FinitePlays^\MaxPl}

\newcommand\play{\ensuremath{\rho}\xspace}

\newcommand\strat{\ensuremath{\sigma}\xspace}

\newcommand\stratmin{\ensuremath{\strat_{\MinPl}}\xspace}
\newcommand\minstrategy{\stratmin}
\newcommand\stratmax{\ensuremath{\strat_{\MaxPl}}\xspace}
\newcommand\maxstrategy{\stratmax}

\newcommand\Play{\ensuremath{\mathsf{Play}}\xspace}
\newcommand\outcomes{\Play}

\newcommand\Val{\ensuremath{\mathsf{Val}}\xspace}
\newcommand\uVal{\ensuremath{\overline{\Val}}\xspace}
\newcommand\uppervalue{\uVal}
\newcommand\lVal{\ensuremath{\underline{\Val}}\xspace}
\newcommand\lowervalue{\lVal}

\newcommand\Clocks{\ensuremath{X}\xspace}
\newcommand\val{\ensuremath{\nu}\xspace}
\newcommand\valnull{\ensuremath{\mathbf{0}}\xspace}
\newcommand\Guards[1]{\ensuremath{G(#1)}\xspace}

\newcommand\States{\ensuremath{S}\xspace}
\newcommand\StatesMin{\ensuremath{\States_{\MinPl}}\xspace}
\newcommand\StatesMax{\ensuremath{\States_{\MaxPl}}\xspace}
\newcommand\StatesT{\ensuremath{\States_t}\xspace}
\newcommand\state{\ensuremath{s}\xspace}

\newcommand\Trans{\ensuremath{\Delta}\xspace}
\newcommand\trans{\ensuremath{\delta}\xspace}

\newcommand\rgame{\ensuremath{\mathcal{R}(\game)}\xspace} 
\newcommand\regions[2]{\mathsf{Reg}(#1,#2)}

\newcommand\hgame{\ensuremath{\mathcal S}(\game)\xspace}

\newcommand\rpath{\ensuremath{\pi}\xspace}

\newcommand\regv{\ensuremath{v}\xspace}
\newcommand\FOG{\ensuremath{\mathsf{FOG}}\xspace}
\newcommand\Ball{\ensuremath{\mathcal{B}_\infty}\xspace}

\newcommand\IteOpe{\ensuremath{\mathcal{F}}\xspace}

\renewcommand\paragraph[1]{\smallskip\noindent\textbf{#1.}}

\title{Optimal Reachability in Divergent Weighted Timed
    Games\thanks{The first author has been supported by ENS Cachan,
    Universit\'e Paris-Saclay. This work has been funded by the DeLTA
    project (ANR-16-CE40-0007), and by the SoSI project (PEPS SISC
    CNRS).}}

\author{Damien~Busatto-Gaston, Benjamin~Monmege and
  Pierre-Alain~Reynier} 
\institute{Aix Marseille Univ, LIF, CNRS, France\\
\email{\{damien.busatto,benjamin.monmege,pierre-alain.reynier\}@lif.univ-mrs.fr}}

\begin{document}

\pagestyle{plain}

\maketitle

\begin{abstract}
  Weighted timed games are played by two players on a timed automaton
  equipped with weights: one player wants to minimise the accumulated
  weight while reaching a target, while the other has an opposite
  objective. Used in a reactive synthesis perspective, this
  quantitative extension of timed games allows one to measure the
  quality of controllers. Weighted timed games are notoriously
  difficult and quickly undecidable, even when restricted to
  non-negative weights. Decidability results exist for subclasses of
  one-clock games, and for a subclass with non-negative weights
  defined by a semantical restriction on the weights of cycles. In
  this work, we introduce the class of \emph{divergent weighted timed
    games} as a generalisation of this semantical restriction to
  arbitrary weights. We show how to compute their optimal value,
  yielding the first decidable class of weighted timed games with
  negative weights and an arbitrary number of clocks. In addition, we
  prove that divergence can be decided in polynomial space. Last, we
  prove that for untimed games, this restriction yields a class of
  games for
  which the value can be computed in polynomial time.
\end{abstract}

\section{Introduction}

Developing programs that verify real-time specifications is
notoriously difficult, because such programs must take care of
delicate timing issues, and are difficult to debug a posteriori. One
research direction to ease the design of real-time software is to
automatise the process. We model the situation into a timed game,
played by a \emph{controller} and an antagonistic \emph{environment}:
they act, in a turn-based fashion, over a \emph{timed
  automaton}~\cite{AD94}, namely a finite automaton equipped with
real-valued variables, called clocks, evolving with a uniform rate. A
usual objective for the controller is to reach a target. We are thus
looking for a \emph{strategy} of the controller, that is a recipe
dictating how to play (timing delays and transitions to follow), so
that the target is reached no matter how the environment plays.
Reachability timed games are decidable~\cite{AsaMal99},
and \EXP-complete~\cite{JurTri07}.

If the controller has a winning strategy in a given reachability timed
game, several such winning strategies could exist. Weighted extensions
of these games have been considered in order to measure the quality of
the winning strategy for the controller~\cite{BCFL04,ABM04}. This
means that the game now takes place over a \emph{weighted (or priced)
  timed automaton}~\cite{BehFeh01,AluLa-04}, where transitions are
equipped with weights, and states with rates of weights (the cost is
then proportional to the time spent in this state, with the rate as
proportional coefficient). While solving weighted timed automata has
been shown to be \PSPACE-complete~\cite{BouBri07} (i.e.\ the same
complexity as the non-weighted version), weighted timed games are
known to be undecidable~\cite{BBR05}. This has led to many
restrictions in order to regain decidability, the first and most
interesting one being the class of strictly non-Zeno cost with only
non-negative weights (in transitions and states)~\cite{BCFL04,ABM04}:
this hypothesis states that every execution of the timed automaton
that follows a cycle of the region automaton has a weight far from 0
(in interval $[1,+\infty)$, for instance).

Less is known for weighted timed games in the presence of negative
weights in transitions and/or states. In particular, no results exist
so far for a class that does not restrict the number of clocks of the
timed automaton to 1. However, negative weights are particularly
interesting from a modelling perspective, for instance in case weights
represent the consumption level of a resource (money, energy\dots)
with the possibility to spend and gain some resource. In this work, we
introduce a generalisation of the strictly non-Zeno cost hypothesis in
the presence of negative weights, that we call \emph{divergence}. We
show the decidability of the class of divergent weighted timed games,
with a $2$-\EXP\ complexity (and an $\EXP$-hardness lower
bound). These complexity results match the ones that could be obtained
in the non-negative case from the study of~\cite{BCFL04,ABM04}.

Other types of payoffs than the accumulated weight we study (i.e.\
total payoff) have been considered for weighted timed games. For
instance, energy and mean-payoff timed games have been introduced
in~\cite{BreCas14}. They are also undecidable in
general. Interestingly, a subclass called \emph{robust timed games},
not far from our divergence hypothesis, admits decidability results. A
weighted timed game is robust if, to say short, every simple cycle
(cycle without repetition of a state) has weight non-negative or less
than a constant $-\varepsilon$. Solving robust timed game can be done
in \EXPSPACE, and is \EXP-hard.  Moreover, deciding if a weighted
timed game is robust has complexity $2$-\EXPSPACE\ (and
$\coNEXP$-hard). In contrast, we show that deciding the divergence of
a weighted timed game is a $\PSPACE$-complete
problem.\footnote{Whereas all divergent weighted game are robust, the
  converse may not be true, since it is possible to mix positive and
  negative simple cycles in an SCC.} In terms of modeling power, we do
believe that divergence is sufficient for most cases. It has to be
noted that extending our techniques and results in the case of robust
timed games is intrinsically not possible: indeed, the value problem
for this class is undecidable~\cite{BJM15}.

The property of divergence is also interesting in the absence of time.
Indeed, weighted games with reachability objectives have been recently
explored as a refinement of mean-payoff games~\cite{BGHM15,BGHM16}. A
pseudo-polynomial time (i.e.\ polynomial if weights are encoded in
unary) procedure has been proposed to solve them, and they are at
least as hard as mean-payoff games. In this article, we also study
divergent weighted games, and show that they are the first non-trivial
class of weighted games with negative weights solvable in polynomial
time. Table~\ref{tab:summary} summarises our results. We start in
Sections~\ref{sec:weighted-games} and
\ref{sec:solving-divergent-weighted-games} by studying weighted
(untimed) games, before considering the timed setting in
Sections~\ref{sec:weighted-timed-games} and
\ref{sec:solving-divergent-weighted-timed-games}.

\begin{table}[tbp]
  \centering
  \caption{Deciding weighted (timed) games with arbitrary weights}
  \label{tab:summary}
  \begin{tabular}{|c|c|c|c|}
    \hhline{~|*{3}{-|}}
    \multicolumn{1}{c|}{} & {\cellcolor{gray!20}}Value of a game 
    & \cellcolor{gray!20}Value of a divergent game 
    & \cellcolor{gray!20}Deciding the divergence \\
    \hline
    \cellcolor{gray!20}Untimed & pseudo-poly. \cite{BGHM16} &
                                                              \P-complete
    & \NL-complete (unary), \P (binary) \\\hline
    \cellcolor{gray!20}Timed & Undecidable \cite{BBR05} & 2-\EXP, \EXP-hard & \PSPACE-complete\\\hline
  \end{tabular}
\end{table}

\section{Weighted games}\label{sec:weighted-games}

We start our study with untimed games.
We consider two-player turn-based games played on weighted graphs and
denote the two \emph{players} by $\MaxPl$ and $\MinPl$. A
\emph{weighted game}\footnote{Weighted games are called \emph{min-cost
    reachability games} in \cite{BGHM16}.} is a tuple
$\game=\langle \Vertices=\VerticesMin\uplus
\VerticesMax,\VerticesT,\alphabet, \Edges,\weight\rangle$ where
$\Vertices$ are vertices, partitioned into vertices belonging to
\MinPl (\VerticesMin) and \MaxPl (\VerticesMax),
$\VerticesT\subseteq \VerticesMin$ is a subset of target vertices for
player $\MinPl$, $\alphabet$ is an alphabet,
$\Edges\subseteq \Vertices\times\alphabet\times \Vertices$ is a set of
directed edges, and $\weight\colon \Edges \to \Z$ is the weight
function, associating an integer weight with each edge. These games
need not be finite in general, but in
Sections~\ref{sec:weighted-games} and
\ref{sec:solving-divergent-weighted-games}, we limit our study to the
resolution of finite weighted games (where all previous sets are
finite). We suppose that:
\begin{inparaenum}[($i$)]
\item the game is deadlock-free, i.e.\ for each vertex
  $v\in \Vertices$, there is a letter $a\in A$ and a vertex
  $v'\in \Vertices$, such that $(v,a,v')\in \Edges$;
\item the game is deterministic, i.e.\ for each pair
  $(v,a)\in \Vertices\times \alphabet$, there is at most one vertex
  $v'\in\Vertices$ such that $(v,a,v')\in\Edges$.\footnote{Actions are
    not standardly considered, but they become useful in the
    timed setting.}
\end{inparaenum}

A \emph{finite play} is a finite sequence of edges
$\play=v_0\xrightarrow{a_0}v_1\xrightarrow{a_1}\cdots
\xrightarrow{a_{k-1}}v_k$, i.e.\ for all $0\leq i<k$,
$(v_i,a_i,v_{i+1})\in \Edges$. We denote by $|\play|$ the length $k$
of $\play$. We often write $v_0\xrightarrow{\play}v_k$ to denote that
\play is a finite play from $v_0$ to $v_k$. The play \play is said to
be a \emph{cycle} if $v_k=v_0$. We let $\FinitePlays_\game$ be the set
of all finite plays in $\game$, whereas $\FinitePlaysMin_\game$ and
$\FinitePlaysMax_\game$ denote the finite plays that end in a vertex
of $\MinPl$ and $\MaxPl$, respectively. A \emph{play} is then an
infinite sequence of consecutive edges.

A \emph{strategy} for $\MinPl$ (respectively, $\MaxPl$) is a mapping
$\strat\colon \FinitePlaysMin_\game \to \alphabet$ (respectively,
$\strat\colon \FinitePlaysMax_\game \to \alphabet$) such that for all
finite plays $\play\in\FinitePlaysMin_\game$ (respectively,
$\play\in\FinitePlaysMax_\game$) ending in vertex $v_k$, there exists
a vertex $v'\in\Vertices$ such that
$(v_k,\strat(\play),v')\in \Edges$. A play or finite play
$\play = v_0\xrightarrow{a_0}v_1\xrightarrow{a_1}\cdots$ conforms to a
strategy $\strat$ of $\MinPl$ (respectively, $\MaxPl$) if for all $k$
such that $v_k\in \VerticesMin$ (respectively, $v_k\in\VerticesMax$),
we have that $a_{k} = \strat(v_0\xrightarrow{a_0}v_1\cdots v_k)$. A
strategy $\strat$ is \emph{memoryless} if for all finite plays
$\play, \play'$ ending in the same vertex, we have that
$\strat(\play)=\strat(\play')$.  For all strategies $\minstrategy$ and
$\maxstrategy$ of players \MinPl and \MaxPl, respectively, and for all
vertices~$v$, we let $\outcomes_\game(v,\maxstrategy,\minstrategy)$ be
the outcome of $\maxstrategy$ and $\minstrategy$, defined as the
unique play conforming to $\maxstrategy$ and $\minstrategy$ and
starting in~$v$.

The objective of \MinPl is to reach a target vertex, while minimising
the accumulated weight up to the target. Hence, we associate to every
finite play
$\play=v_0\xrightarrow{a_0}v_1 \ldots\xrightarrow{a_{k-1}}v_k$ its
accumulated weight
$\weight_\game(\play)=\sum_{i=0}^{k-1}
\weight(v_i,a_i,v_{i+1})$. Then, the weight of an infinite play
$\play=v_0\xrightarrow{a_0}v_1\xrightarrow{a_1}\cdots$, also denoted
by $\weight_\game(\play)$, is defined by $+\infty$ if
$v_k\notin \VerticesT$ for all $k\geq 0$, or the weight of
$v_0\xrightarrow{a_0}v_1 \ldots\xrightarrow{a_{k-1}}v_k$ if $k$ is the
first index such that $v_k\in\VerticesT$. Then, we let
$\Val_\game(v,\minstrategy)$ and $\Val_\game(v,\maxstrategy)$ be the
respective values of the strategies:
\begin{align*}
\Val_\game(v,\minstrategy) &= \sup_{\maxstrategy}
\weight_\game(\outcomes(v,\maxstrategy,\minstrategy)) \\
\Val_\game(v,\maxstrategy) &= \inf_{\minstrategy}
\weight_\game(\outcomes(v,\maxstrategy,\minstrategy))\,.
\end{align*}
Finally, for all vertices~$v$, we let
$\lowervalue_\game(v) = \sup_{\maxstrategy}
\Val_\game(v,\maxstrategy)$ and
$\uppervalue_\game(v) = \inf_{\minstrategy}
\Val_\game(v,\minstrategy)$ be the \emph{lower} and \emph{upper
  values} of $v$, respectively.  We may easily show that
$\lowervalue_\game(v)\leq \uppervalue_\game(v)$ for all~$v$. We say
that strategies $\minstrategy^\star$ of $\MinPl$ and
$\maxstrategy^\star$ of $\MaxPl$ are optimal if, for all vertices~$v$,
$\Val_\game(v,\maxstrategy^\star)=\lowervalue_\game(v)$ and
$\Val_\game(v,\minstrategy^\star)=\uppervalue_\game(v)$, respectively.
We say that a game $\game$ is \emph{determined} if for all
vertices~$v$, its lower and upper values are equal. In that case, we
write $\Val_\game(v)=\lowervalue_\game(v)=\uppervalue_\game(v)$, and
refer to it as the \emph{value} of~$v$ in $\game$. Finite weighted
games are known to be determined~\cite{BGHM16}.  If the game is clear
from the context, we may drop the index $\game$ from all previous
notations.

\paragraph{Problems}
We want to compute the value of a \emph{finite} weighted game,
as well as optimal strategies for both players, if they exist. The
corresponding decision problem, called the \emph{value problem}, asks
whether $\Val_\game(v) \leq \alpha$, given a finite weighted game
$\game$, one of its vertices $v$, and a threshold
$\alpha\in\Z\cup\{-\infty,+\infty\}$.

\paragraph{Related work}
The value problem is a generalisation of the classical shortest path
problem in a weighted graph to the case of two-player games. If
weights of edges are all non-negative, a generalised Dijkstra
algorithm enables to solve it in polynomial time~\cite{KBB+08}. In the
presence of negative weights, a pseudo-polynomial-time (i.e.\
polynomial with respect to the game where weights are stored in unary)
solution has been given in~\cite{BGHM16}, based on a fixed point
computation with value iteration techniques. Moreover, the value
problem with threshold $-\infty$ is shown to be in $\NP\cap \coNP$,
and as hard as solving mean-payoff games.

\section{Solving divergent weighted games}
\label{sec:solving-divergent-weighted-games}

Our first contribution is to solve in polynomial time the value
problem, for a subclass of finite weighted games that we call
\emph{divergent}. To the best of our knowledge, this is the first
attempt to solve a non-trivial class of weighted games with arbitrary
weights in polynomial time. Moreover, the same core technique is used
for the decidability result in the timed setting that we will present
in the next sections. Let us first define the class of divergent
weighted games:

\begin{definition}
  A weighted game \game is divergent when every cycle \play of \game
  satisfies $\weight(\play)\neq 0$.
\end{definition}

Divergence is a property of the underlying weighted graph, independent
from the repartition of vertices between players. The term
\emph{divergent} reflects that cycling in the game ultimately makes
the accumulated weight grow in absolute value. We will first formalise
this intuition by analysing the strongly connected components (SCC) of
the graph structure of a divergent game (the repartition of vertices
into players does not matter for the SCC decomposition). Based on this
analysis, we will obtain the following results:
\begin{theorem}
  The value problem over finite divergent weighted games is a
  \P-complete problem. Moreover, deciding if a given finite weighted
  game is divergent is an \NL-complete problem when weights are
  encoded in unary, and \P\ when they are encoded in binary.
\end{theorem}

\paragraph{SCC analysis}
A play \play in $\game$ is said to be positive (respectively,
negative) if $\weight(\play)>0$ (respectively, $\weight(\play)<0$).
It follows that a cycle in a divergent weighted game is either
positive or negative. A cycle is said to be simple if no vertices are
visited twice (except for the common vertex at the beginning and the
end of the cycle). We will rely on the following characterisation of
divergent games in terms of SCCs.

\begin{proposition}\label{prop:scc-sign}
  A weighted game \game is divergent if and only if, in each SCC
  of~\game, all simple cycles are either all positive, or all
  negative.
\end{proposition}
\begin{proof}
  Let us first suppose that \game is divergent. By contradiction,
  consider a negative simple cycle \play (of weight $-p<0$) and a
  positive simple cycle $\play'$ (of weight $p'>0$) in the same SCC.
  Let $v$ and $v'$ be respectively the first vertices of \play and
  $\play'$. By strong connectivity, there exists a finite play $\eta$
  from $v$ to $v'$ and a finite play $\eta'$ from $v'$ to $v$.  Let us
  consider the cycle $\play''$ obtained as the concatenation of $\eta$
  and $\eta'$. If $\play''$ has weight $q>0$, the cycle obtained by
  concatenating $q$ times \play and $p$ times $\play''$ has weight
  $0$, which contradicts the divergence of~$\game$. The same reasoning
  on $\play''$ and $\play'$ proves that $\play''$ can not be
  negative. Thus, $\play''$ is a cycle of weight $0$, which again
  contradicts the hypothesis.

  Reciprocally, consider a cycle of \game. It can be decomposed into
  simple cycles, all belonging to the same SCC. Therefore they are all
  positive or all negative. As the accumulated weight of the cycle is
  the sum of the weights of these simple cycles, \game is
  divergent.\qed
\end{proof}

\paragraph{Computing the values} 
Consider a divergent weighted game $\game$. Let us start by observing
that vertices with value $+\infty$ are those from which \MinPl can not
reach the target vertices: thus, they can be computed with the
classical attractor algorithm, and we can safely remove them, without
changing other values or optimal strategies. In the rest, we therefore
assume all values to be in $\Z\cup\{-\infty\}$.

Our computation of the values relies on a value iteration algorithm to
find the greatest fixed point of operator
$\IteOpe\colon (\Z\cup\{-\infty,+\infty\})^\Vertices \to
(\Z\cup\{-\infty,+\infty\})^\Vertices$, defined for every vector
$\vec x$ by $\IteOpe(\vec x)_v = 0$ if $v\in \VerticesT$, and otherwise
\[\IteOpe(\vec x)_v =
  \begin{cases}
    \displaystyle{\min_{e = (v, a, v')\in\Edges} \weight(e) + \vec x_{v'}} &
    \text{if }
    v\in \VerticesMin\\
    \displaystyle{\max_{e = (v, a, v')\in\Edges} \weight(e) + \vec x_{v'}} &
    \text{if } v\in \VerticesMax\,.
  \end{cases}
\]
Indeed, this greatest fixed point is known to be the vector of values
of the game (see, e.g., \cite[Corollary~11]{BGHM16}). In
\cite{BGHM16}, it is shown that, by initialising the iterative
evaluation of \IteOpe with the vector $\vec x^0$ mapping all vertices
to $+\infty$, the computation terminates after a number of iterations
pseudo-polynomial in $\game$ (i.e.\ polynomial in the number of
vertices and the greatest weight in $\game$). For $i>0$, we let
$\vec x^i = \IteOpe(\vec x^{i-1})$. Notice that the sequence
$(\vec x^i)_{i\in\N}$ is non-increasing, since \IteOpe is a monotonic
operator.
Value iteration algorithms usually benefit from decomposing a game
into SCCs (in polynomial time), considering them in a bottom-up
fashion: starting with target vertices that have value $0$, SCCs are
then considered in inverse topological order since the values of
vertices in an SCC only depend on values of vertices of greater SCCs
(in topological order), that have been previously computed.

\begin{example}
  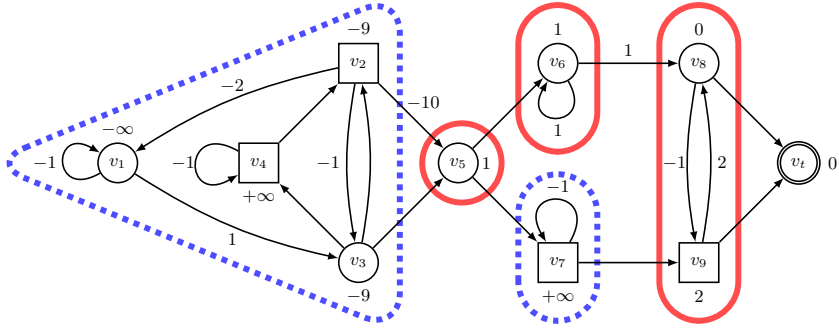
\begin{figure}[tbp]
    \centering \scalebox{.75}{\begin{tikzpicture}[node
        distance=2.5cm,auto,->,>=latex]

        \draw[dashed,draw=blue!70,rounded corners=20pt,line width=3pt]
        (-2.3,0) -- (5,3) -- (5,-3) -- cycle;
        \draw[dashed,draw=blue!70,rounded corners=20pt,line width=3pt]
        (7.1,-2.8) rectangle (8.5,-.2);
        \draw[draw=red!70,rounded corners=20pt,line width=3pt]
        (7.1,2.8) rectangle (8.5,.2);
        \draw[draw=red!70,rounded corners=20pt,line width=3pt]
        (9.6,-2.8) rectangle (11,2.8);
        \node[draw=red!70,rounded corners=20pt,line width=3pt,inner
        sep = 20pt] () at (6.1,0) {};

        \node[player1](1){\makebox[0mm][c]{$\vertex_1$}};
        \node[player2](4)[right of=1]{\makebox[0mm][c]{$\vertex_4$}};
        \node[player2](2)[above right of=4]{\makebox[0mm][c]{$\vertex_2$}}; 
        \node[player1](3)[below right of=4]{\makebox[0mm][c]{$\vertex_3$}};
        \node[player1](5)[below right of=2]{\makebox[0mm][c]{$\vertex_5$}};
        \node[player1](6)[above right of=5]{\makebox[0mm][c]{$\vertex_6$}};
        \node[player2](7)[below right of=5]{\makebox[0mm][c]{$\vertex_7$}}; 
        \node[player1](8)[right of=6]{\makebox[0mm][c]{$\vertex_8$}}; 
        \node[player2](9)[right of=7]{\makebox[0mm][c]{$\vertex_9$}};
        \node[player1](10)[accepting,below right of=8]{\makebox[0mm][c]{$\vertex_t$}};

        \node()[above of=1,node distance=6mm]{$-\infty$}; 
        \node()[above of=2,node distance=6mm]{$-9$}; 
        \node()[below of=3,node distance=6mm]{$-9$}; 
        \node()[below of=4,node distance=6mm]{$+\infty$};
        \node()[right of=5,node distance=5mm]{$1$}; 
        \node()[above of=6,node distance=6mm]{$1$}; 
        \node()[below of=7,node distance=6mm]{$+\infty$}; 
        \node()[above of=8,node distance=6mm]{$0$};
        \node()[below of=9,node distance=6mm]{$2$}; 
        \node()[right of=10,node distance=6mm]{$0$};

        \path (1) edge[in=150,out=210,loop]
        node[left]{$-1$} (1) 
        (1) edge[bend right=10] node[below]{$1$} (3)
        (2) edge[bend right=10] node[above]{$-2$} (1) 
        (2) edge[bend right=10] node[left]{$-1$} (3) 
        (2) edge node[above right,xshift=-2mm]{$-10$} (5) 
        (3) edge[bend right=10] (2)
        (3) edge (4) 
        (3) edge (5) 
        (4) edge[in=210,out=150,loop] node[left,xshift=1mm]{$-1$} (4) 
        (4) edge (2)
        (5) edge (6)
        (5) edge (7)
        (6) edge[in=-120,out=-60,loop] node[below]{$1$} (6)
        (6) edge node[above]{$1$} (8) 
        (7) edge[in=-240,out=-300,loop] node[above]{$-1$} (7)
        (7) edge (9) 
        (8) edge[bend right=10] node[left,xshift=1mm]{$-1$} (9)
        (8) edge (10) 
        (9) edge[bend right=10] node[right]{$2$} (8) 
        (9) edge (10);

\end{tikzpicture}}
\caption{SCC decomposition of a divergent weighted game:
  $\{v_1,v_2,v_3,v_4\}$ and $\{v_7\}$ are negative SCCs, $\{v_6\}$ and
  $\{v_8,v_9\}$ are positive SCCs, and $\{v_5\}$ is a trivial positive
  SCC.}
\label{fig:SCC}
\end{figure}
Consider the weighted game of \figurename~\ref{fig:SCC}, where \MinPl
vertices are drawn with circles, and \MaxPl vertices with
squares. Vertex $\vertex_t$ is the only target. Near each vertex is
placed its value. For a given vector $\vec x$, we have
$\IteOpe(\vec x)_{v_8} = \min(0+\vec x_{v_t},-1+\vec x_{v_9})$ and
$\IteOpe(\vec x)_{v_2} = \max(-2+\vec x_{v_1},-1+\vec x_{v_3},-10+\vec
x_{v_5})$. By a computation of the attractor of $\{v_t\}$ for \MinPl,
we obtain directly that $v_4$ and $v_7$ have value~$+\infty$. The
inverse topological order on SCCs prescribes then to compute first the
values for the SCC $\{v_8,v_9\}$, with target vertex $v_t$ associated
with value~$0$. Then, we continue with SCC $\{v_6\}$, also keeping a
new target vertex $v_8$ with (already computed) value $0$. For the
trivial SCC $\{v_5\}$, a single application of $\IteOpe$ suffices to
compute the value. Finally, for the SCC $\{v_1,v_2,v_3,v_4\}$, we keep
a new target vertex $v_5$ with value $1$.\footnote{This means that, in
  the definition of $\IteOpe$, a vertex $v$ of $\VerticesT$ is indeed
  mapped to its previously computed value, not necessarily $0$.}
Notice that this game is divergent, since, in each SCC, all simple
cycles have the same sign.
\end{example}

For a divergent game $\game$, Proposition~\ref{prop:scc-sign} allows
us to know in polynomial time if a given SCC is positive or negative,
i.e.\ if all cycles it contains are positive or negative,
respectively: it suffices to consider an arbitrary cycle of it, and
compute its weight. A trivial SCC (i.e.\ with a single vertex and no
edges) will be arbitrarily considered positive. We now explain how to
compute in polynomial time the value of all vertices in a positive or
negative SCC.

First, in case of a positive SCC, we show that: 
\begin{proposition}\label{prop:positive-scc}
  The value iteration algorithm applied on a positive SCC with $n$
  vertices stabilises after at most $n$ steps.
\end{proposition}
\begin{proof}[inspired by techniques used in \cite{BCFL04}] Let
  $W=\max_{\edge\in\Edges} |\weight(\edge)|$ be the greatest weight in
  the game. There are no negative cycles in the SCC, thus there are no
  vertices with value $-\infty$ in the SCC, and all values are
  finite. Let $K$ be an upper bound on the values $|\vec x^n_\vertex|$
  obtained after $n$ steps of the algorithm.\footnote{After $n$ steps,
    the value iteration algorithm has set to a finite value all
    vertices, since it extends the attractor computation.} Fix an
  integer $p>(2K+W(n-1))n$. We will show that the values obtained
  after $n+p$ steps are identical to those obtained after $n$ steps
  only. Therefore, since the algorithm computes non-increasing
  sequences of values, we have indeed stabilised after $n$ steps only.
  Assume the existence of a vertex $\vertex$ such that
  $\vec x^{n+p}_\vertex<\vec x^{n}_\vertex$. By induction on $p$, we
  can show (see Lemma~\ref{lem:valite} in Appendix~\ref{app:technical}
  for a detailed proof) the existence of a vertex $\vertex'$ and a
  finite play \play from \vertex to $\vertex'$ with length $p$ and
  weight $\vec x^{n+p}_\vertex-\vec x^{n}_{\vertex'}$: the play is
  composed of the edges that optimise successively the min/max
  operator in \IteOpe. This finite play being of length greater than
  $(2K+W(n-1))n$, there is at least one vertex appearing more than
  $2K+W(n-1)$ times. Thus, it can be decomposed into at least
  $2K+W(n-1)$ cycles and a finite play $\play'$ visiting each vertex
  at most once. All cycles of the SCC being positive, the weight of
  \play is at least $2K+W(n-1) - (n-1) W= 2K$, bounding from below the
  weight of $\play'$ by $-(n-1)W$. Then,
  $\vec x^{n+p}_\vertex-\vec x^{n}_{\vertex'} \geq 2K$, so
  $\vec x^{n+p}_\vertex \geq 2K + \vec x^{n}_{\vertex'} \geq K$.  But
  $K \geq \vec x^{n}_{\vertex}$, so
  $\vec x^{n+p}_\vertex\geq\vec x^{n}_{\vertex}$, and that is a
  contradiction.\qed
\end{proof}

\begin{example}
  For the SCC $\{v_8,v_9\}$ of the game in \figurename~\ref{fig:SCC},
  starting from $\vec x$ mapping $v_8$ and $v_9$ to $+\infty$, and
  $v_t$ to $0$, after one iteration, $\vec x_{v_8}$ changes for value
  $0$, and after the second iteration, $\vec x_{v_9}$ stabilises to
  value $2$.
\end{example}

Consider then the case of a negative SCC. Contrary to the previous
case, we must deal with vertices of value $-\infty$. However, in a
negative SCC, those vertices are easy to find\footnote{This is in
  contrast with the general case of (non divergent) finite weighted
  games where the problem of deciding if a vertex has value $-\infty$
  is as hard as solving mean-payoff games~\cite{BGHM16}.}. These are
all vertices where $\MaxPl$ can not unilaterally guarantee to reach a
target vertex:

\begin{proposition}\label{prop:minus-infty}
  In a negative SCC with no vertices of value $+\infty$, vertices of
  value $-\infty$ are all the ones not in the attractor for $\MaxPl$
  to the targets.
\end{proposition}
\begin{proof}
  Consider a vertex $v$ in the attractor for $\MaxPl$ to the
  targets. Then, if \MaxPl applies a winning memoryless strategy for
  the reachability objective to the target vertices, all strategies of
  \MinPl will generate a play from $v$ reaching a target after at most
  $|\Vertices|$ steps. This implies that $v$ has a finite (lower)
  value in the game.

  Reciprocally, if $v$ is not in the attractor, by determinacy of
  games with reachability objectives, \MinPl has a (memoryless)
  strategy $\minstrategy$ to ensure that no strategy of \MaxPl permits
  to reach a target vertex from $v$. Applying $\minstrategy$ long
  enough to generate many negative cycles, before switching to a
  strategy allowing \MinPl to reach the target (such a strategy exists
  since no vertex has value $+\infty$ in the game), allows \MinPl to
  obtain from $v$ a negative weight as small as possible. Thus, $v$
  has value~$-\infty$. \qed
\end{proof}

Thus, we can compute vertices of value $-\infty$ in polynomial time
for a negative SCC. Then, finite values of other vertices can be
computed in polynomial time with the following procedure. From a
negative SCC $\game$ that has no more vertices of value $+\infty$ or
$-\infty$, consider the dual (positive) SCC $\widetilde \game$ obtained by:
\begin{inparaenum}[($i$)]
\item switching vertices of $\MinPl$ and $\MaxPl$; 
\item taking the opposite of every weight in edges.
\end{inparaenum}
Sets of strategies of both players are exchanged in those two games,
so that the upper value in $\game$ is equal to the opposite of the
lower value in $\widetilde \game$, and vice versa. Since weighted
games are determined, the value of $\game$ is the opposite of the
value of $\widetilde\game$. Then, the value of $\game$ can be deduced
from the value of $\widetilde\game$, for which
Proposition~\ref{prop:positive-scc} applies. We may also interpret
this result as follows: 
\begin{proposition}\label{prop:negative-scc}
  The value iteration algorithm, initialised with $\vec x^0_v=-\infty$
  (for all $v$), applied on a negative SCC with $n$ vertices, and no
  vertices of value $+\infty$ or $-\infty$, stabilises after at most
  $n$ steps.
\end{proposition}
\begin{proof}
  It is immediate that the vectors computed with this modified value
  iteration (that computes the smallest fixed point of $\IteOpe$) are
  exactly the opposite vectors of the ones computed in the dual
  positive SCC. The previous explanation is then a justification of
  the result.\qed
\end{proof}

\begin{example}
  Consider the SCC $\{v_1,v_2,v_3,v_4\}$ of the game in
  \figurename~\ref{fig:SCC}, where the value of vertex $v_5$ has been
  previously computed. We already know that $v_4$ has value $+\infty$
  so we do not consider it further. The attractor of $\{v_5\}$ for
  \MaxPl is $\{v_2,v_3\}$, so that the value of $v_1$ is
  $-\infty$. Then, starting from $\vec x_0$ mapping $v_2$ and $v_3$ to
  $-\infty$, the value iteration algorithm computes this sequence of
  vectors: $\vec x_1 = (-9,-\infty)$ (\MaxPl tries to maximise the
  payoff, so he prefers to jump to the target to obtain $-10+1$ than
  going to $v_3$ where he gets $-1-\infty$, while \MinPl chooses $v_2$
  to still guarantee $0-\infty$), $\vec x_2 = (-9,-9)$ (now, \MinPl
  has a choice between the target giving $0+1$ or $v_3$ giving $0-9$).
\end{example}

The proof for \P-hardness comes from a reduction (in logarithmic
space) of the problem of solving finite games with reachability
objectives \cite{Imm81}. To a reachability game, we simply add weights
1 on every transition, making it a divergent weighted game. Then,
\MinPl wins the reachability game if and only if the value in the
weighted game is lower than $|\Vertices|$.

In a divergent weighted game where all values are finite, optimal
strategies exist. As observed in \cite{BGHM16}, \MaxPl always has a
memoryless optimal strategy, whereas \MinPl may require (finite)
memory. Optimal strategies for both players can be obtained by
combining optimal strategies in each SCC, the latter being obtained as
explained in \cite{BGHM16}. 

\paragraph{Class decision} We explain why deciding the divergence of a
weighted game is an $\NL$-complete problem when weighs are encoded in
unary. First, to prove the membership in \NL, notice that a weighted
game is \emph{not divergent} if and only if there is a positive cycle
and a negative cycle, both of length at most $|\Vertices|$, and
belonging to the same SCC.\footnote{If the game is not divergent,
  there exists an SCC containing a negative simple cycle and a
  positive one by Proposition~\ref{prop:scc-sign}. This implies the
  existence of a negative cycle and a positive cycle in the same SCC,
  both of length at most $|\Vertices|$. Reciprocally, this property
  implies the non-divergence, by the same proof as for
  Proposition~\ref{prop:scc-sign}.}  To test this property in \NL, we
first guess a starting vertex for both cycles. Verifying that those
are in the same SCC can be done in \NL. Then, we guess the two cycles
on-the-fly, keeping in memory their accumulated weights (smaller than
$W\times |\Vertices|$, with $W$ the biggest weight in the game, and
thus of size at most logarithmic in the size of $\game$, if weights
are encoded in unary), and stop the on-the-fly exploration when the
length of the cycles exceeds $|\Vertices|$. Therefore testing
divergence is in $\co\NL=\NL$ \cite{Imm88,Sze88}.

The $\NL$-hardness (indeed $\co\NL$-hardness, which is equivalent
\cite{Imm88,Sze88}) is shown by a reduction of the reachability
problem in a finite automaton. More precisely, we consider a finite
automaton with a starting state and a different target state without
outgoing transitions. We construct from it a weighted game by
distributing all states to \MinPl, and equipping all transitions with
weight $1$. We also add a loop with weight $-1$ on the target state
and a transition from the target state to the initial state with
weight $0$. Then, the game is not divergent if and only if the target
can be reached from the initial state in the automaton.

When weights are encoded in binary, the previous decision procedure
gives \NP\ membership. However, we can achieve a \P\ upperbound
with the following procedure. For every vertex~$\vertex$, let
$C_\vertex=\{\weight(\play) \mid \play\text{ cycle containing
}\vertex\}$. Using Floyd-Warshall's algorithm, it is possible to
compute in polynomial time $\inf C_\vertex$ (in particular, it detects
if $C_\vertex\neq\emptyset$), as well as $\sup C_\vertex$ in a dual
fashion. Then, Proposition~\ref{prop:scc-sign} allows us to guarantee
that a weighted game $\game$ is divergent if and only if
$0\not\in[\inf C_\vertex,\sup C_\vertex]$, for all vertices $\vertex$
such that $C_\vertex\neq\emptyset$.

\section{Weighted timed games}\label{sec:weighted-timed-games}

We now turn our attention to a timed extension of the weighted games. We
will first define weighted timed games, giving their semantics in
terms of \emph{infinite} weighted games.
We let \Clocks be a finite set of variables called clocks. A valuation of
clocks is a mapping $\val\colon \Clocks\to \Rplus$. For a valuation $\val$,
$d\in\Rplus$ and $Y\subseteq \Clocks$, we define the valuation $\val+d$ as
$(\val+d)(x)=\val(x)+d$, for all $x\in \Clocks$, and the valuation
$\val[Y\leftarrow 0]$ as $(\val[Y\leftarrow 0])(x)=0$ if $x\in Y$, and
$(\val[Y\leftarrow 0])(x)=\val(x)$ otherwise. The valuation $\valnull$
assigns $0$ to every clock.
A guard on clocks of \Clocks is a conjunction of atomic constraints of
the form $x\bowtie c$, where ${\bowtie}\in\{{\leq},<,=,>,{\geq}\}$ and
$c\in \N$. A valuation $\val\colon \Clocks\to \Rplus$ satisfies an
atomic constraint $x\bowtie c$ if $\val(x)\bowtie c$. The satisfaction
relation is extended to all guards $g$ naturally, and denoted by
$\val\models g$. We let $\Guards \Clocks$ the set of guards over
\Clocks.

A weighted timed game is then a tuple
$\game=\langle\States=\StatesMin\uplus\StatesMax,\StatesT,
\Trans,\weight\rangle$ where $\StatesMin$ and $\StatesMax$ are
\emph{finite} disjoint subsets of states belonging to \MinPl and
\MaxPl, respectively, $\StatesT\subseteq \StatesMin$ is a subset of
target states for player $\MinPl$,
$\Trans\subseteq \States\times\Guards \Clocks\times \powerset \Clocks
\times \States$ is a \emph{finite} set of transitions, and
$\weight\colon \Trans\uplus\States \to \Z$ is the weight function,
associating an integer weight with each transition and state. Without
loss of generality, we may suppose that for each state
$\state\in \States$ and valuation $\val$, there exists a transition
$(\state,g,Y,\state')\in \Trans$ such that $\val\models g$.

The semantics of a weighted timed game $\game$ is defined in terms of
the infinite weighted game $\mathcal H$ whose vertices are
configurations of the weighted timed game. A configuration is a pair
$(\state,\val)$ with a state and a valuation of the
clocks. Configurations are split into players according to the
state. A configuration is final if its state is final. The alphabet of
$\mathcal H$ is given by $\Rplus\times \Trans$ and will encode the
delay that a player wants to spend in the current state, before firing
a certain transition. For every delay $d\in\Rplus$, transition
$\trans=(\state,g,Y,\state')\in \Trans$ and valuation~$\val$, there is
an edge $(\state,\val)\xrightarrow{d,\trans}(\state',\val')$ if
$\val+d\models g$ and $\val'=(\val+d)[Y\leftarrow 0]$. The weight of
such an edge $e$ is given by
$d\times \weight(\state) + \weight(\trans)$.

Plays, strategies, and values in the weighted timed game $\game$ are
then defined as the ones in $\mathcal H$. It is known that weighted
timed games are determined
($\lowervalue_\game(s,\nu)=\uppervalue_\game(s,\nu)$ for all state $s$
and valuation $\nu$).\footnote{The result is stated in \cite{BGH+15}
  for weighted timed games (called priced timed games) with one clock,
  but the proof does not use the assumption on the number of clocks.}

As usual in related work \cite{ABM04,BCFL04,BJM15}, we assume that all
clocks are \emph{bounded}, i.e.\ there is a constant $M\in\N$ such
that every transition of the weighted timed games is equipped with a
guard $g$ such that $\val\models g$ implies $\val(x)\leq M$ for all
clocks $x\in \Clocks$.
We will rely on the crucial notion of regions, as introduced in the
seminal work on timed automata \cite{AD94}: a region is a set of
valuations, that are all time-abstract bisimilar. 
There is only a finite number of regions and we denote by
$\regions \Clocks M$ the set of regions associated with set of clocks
\Clocks and maximal constant $M$ in guards. For a valuation~$\val$, we
denote by $[\val]$ the region that contains it. A region $r'$ is said
to be a time successor of region $r$ if there exist $\val\in r$,
$\val'\in r'$, and $d>0$ such that $\val'=\val+d$. Moreover, for
$Y\subseteq \Clocks$, we let $r[Y\leftarrow 0]$ be the region where
clocks of $Y$ are reset.

The region automaton $\rgame$ of a game
$\game= \langle\States=\StatesMin\uplus\StatesMax,\StatesT,
\Trans,\weight\rangle$ is the finite automaton with states
$\States\times \regions \Clocks M$, alphabet $\Trans$, and a
transition $(s,r)\xrightarrow{\trans}(s',r')$ labelled by
$\trans=(s,g,Y,s')$ if there exists a region $r''$ time successor of
$r$ such that $r''$ satisfies the guard $g$, and
$r'=r''[Y\leftarrow 0]$. We call \emph{path} an execution (not
necessarily accepting) of this finite automaton, and we denote by
$\rpath$ the paths. A play $\play$ in $\game$ is projected on a
execution $\rpath$ in $\rgame$, by replacing actual valuations by the
regions containing them: we say that $\play$ \emph{follows} path
$\rpath$. It is important to notice that, even if $\rpath$ is a cycle
(i.e.\ starts and ends in the same state of the region automaton),
there may exist plays following it in $\game$ that are not cycles, due
to the fact that regions are sets of valuations.

\paragraph{Problems} As in weighted (untimed) games, we consider the
\emph{value problem}, mimicked from the one in $\mathcal
H$. Precisely, given a weighted timed game $\game$, a configuration
$(s,\val)$ and a threshold $\alpha\in \Z\cup\{-\infty,+\infty\}$, we
want to know whether $\Val_\game(s,\val)\leq \alpha$. In the context of
timed games, optimal strategies may not exist. We generally focus on
$\varepsilon$-optimal strategies, that guarantee
the optimal value, up to a small error $\varepsilon$.

\paragraph{Related work} In the one-player case, computing the
optimal value and an $\varepsilon$-optimal strategy for weighted timed
automata is known to be $\PSPACE$-complete \cite{BouBri07}. In the
two-player case, much work for weighted timed games (also called
priced timed games in the literature) has been achieved in the case of
non-negative weights. In this setting, the value problem is
undecidable \cite{BBR05,BJM15}. To obtain decidability, one
possibility is to limit the number of clocks to 1: then, there is an
exponential-time algorithm to compute the value as well as
$\varepsilon$-optimal strategies \cite{BBM06,Rut11,DueIbs13}, whereas
the problem is only known to be $\P$-hard. The other possibility to
obtain a decidability result~\cite{ABM04,BCFL04} is to enforce a
semantical property of divergence, originally called strictly non-Zeno
cost: it asks that every play following a cycle in the region
automaton has weight at least~$1$.

In the presence of negative weights, undecidability even holds for
weighted timed games with only 2 clocks \cite{BGNK+14} (for the
existence problem asking if a strategy of player \MinPl can guarantee
a given threshold). Only the 1-clock restriction has been studied so
far allowing one to obtain an exponential-time algorithm, under
restrictions on the resets of the clock in cycles \cite{BGH+15}. For
weighted timed games, the strictly non-Zeno cost property has only
been defined and studied in the absence of negative weights
\cite{BCFL04}. As already mentioned in the introduction, the notion is
close, but not equivalent, to the one of robust weighted timed games,
studied for mean-payoff and energy objectives \cite{BreCas14}. In the
next section, we extend the strictly non-Zeno cost property to
negative weights calling it the divergence property, in order to
obtain decidability of a large class of multi-clocks weighted timed
games in the presence of arbitrary weights.

\section{Solving divergent weighted timed
  games}\label{sec:solving-divergent-weighted-timed-games}

We introduce divergent weighted timed games, as an extension of
divergent weighted games to the timed setting.

\begin{definition}
  A weighted timed game \game is divergent when every finite
  play~\play in \game following a cycle in the region automaton
  $\rgame$ satisfies $\weight(\play)\notin (-1,1)$.\footnote{As in
    \cite{BCFL04}, we could replace $(-1,1)$ by $(-\kappa,\kappa)$ to
    define a notion of $\kappa$-divergence. However, since weights and
    guard constraints in weighted timed games are integers, for
    $\kappa\in(0,1)$, a weighted timed game $\game$ is
    $\kappa$-divergent if and only if it is divergent.}
\end{definition}

The weight is not only supposed to be different from $0$, but also far
from~$0$: otherwise, the original intuition on the ultimate growing of
the values of plays would not be fulfilled. If $\game$ has only
non-negative weights on states and transitions, this definition
matches with the \emph{strictly non-Zeno cost} property
of~\cite[Thm.~6]{BCFL04}. Our contributions summarise as follows:

\begin{theorem}
  The value problem over divergent weighted timed games is decidable
  in $2$-\EXP, and is \EXP-hard. Moreover, deciding if a given weighted
  timed game is divergent is a \PSPACE-complete problem.
\end{theorem}

Remember that these complexity results match the ones that can be
obtained from the study of \cite{BCFL04} for non-negative weights.

\paragraph{SCC analysis} Keeping the terminology of the untimed
setting, a cycle~$\rpath$ of~\rgame is said to be positive
(respectively, negative) if every play \play following~$\rpath$
satisfies $\weight(\play)\geq 1$ (respectively,
$\weight(\play)\leq -1$). By definition, every cycle of the region
automaton of a divergent weighted timed game is positive or
negative. Moreover, notice that checking if a cycle $\rpath$ is
positive or negative can be done in polynomial time with respect to
the length of $\rpath$. Indeed, the set
$\{\weight(\play) \mid \play \text{ is a play following } \rpath \}$ is an
interval, as the image of a convex set by an affine function (see
\cite[Sec.~3.2]{BouBri07} for explanation), and the
extremal points of this interval can be computed in polynomial time by
solving a linear problem \cite[Cor.~1]{BouBri07}. We first transfer in
the timed setting the characterisation of divergent games in terms of
SCCs that we relied on in the untimed setting:

\begin{proposition}\label{prop:timed-scc-sign}
  A weighted timed game \game is divergent if and only if, in each SCC
  of \rgame, simple cycles are either all positive, or all negative.
\end{proposition}

The proof of the reciprocal follows the exact same reasoning than for
weighted games (see Proposition~\ref{prop:scc-sign}). For the direct
implication, the situation is more complex: we need to be more careful
in the composition of cycles with each others, and weights in the
timed game are no longer integers, forbidding the arithmetical
reasoning we applied. To help us, we will rely on the corner-point
abstraction introduced in \cite{BouBri08a} to study multi-weighted
timed automata. It consists in adding a weighted information to the
edges $(s,r)\xrightarrow{\delta}(s',r')$ of the region
automaton. Since the weights depend on the exact valuations $\nu$ and
$\nu'$, taken in regions $r$ and $r'$, respectively, the weight of
such an edge in the region automaton is computed for each pair of
\emph{corners} of the regions. Formally, corners of region $r$ are
valuations in $\overline{r} \cap \N^\Clocks$ (where $\overline{r}$
denotes the topological closure of $r$). Since corners do not
necessarily belong to their regions, we must consider a modified
version $\overline\game$ of the game $\game$ where all strict
inequalities of guards have been replaced with non-strict ones. Then,
for a path $\rpath$ in $\rgame$, we denote by $\overline\rpath$ the
equivalent of path $\rpath$ in $\mathcal R(\overline\game)$. 

In the following, our focus is on cycles of the region automaton, so
we only need to consider the aggregation of all the behaviours
following a cycle. Inspired by the \emph{folded orbit graphs} (FOG)
introduced in~\cite{Pur00}, we define the folded orbit graph
$\FOG(\rpath)$ of a cycle
$\rpath=(\state_1,r=r_1) \xrightarrow{\trans_1} (\state_2,r_2)
\xrightarrow{\trans_2} \cdots \xrightarrow{\trans_n} (\state_1,r)$ in
\rgame as a graph whose vertices are corners of region $r$, and that
contains an edge from corner $\regv$ to corner $\regv'$ if there
exists a finite play $\overline{\play}$ in $\overline\game$ from
$(s_1,\regv)$ to $(s_1,\regv')$ following $\overline\rpath$ jumping
from corners to corners\footnote{Notice that if there is a play from
  $(s_1,\regv)$ to $(s_1,\regv')$ in $\overline\game$, there is
  another one that only jumps at corners of regions.}. We fix such a
finite play $\overline{\play}$ arbitrarily and label the edge between
$\regv$ and $\regv'$ in the FOG by this play: it is then denoted by
$\regv\xrightarrow{\overline{\play}}\regv'$. Moreover, since
$\overline\play$ jumps from corners to corners, its weight
$\weight(\overline\play)$ is an integer, conforming to the definitions
of the corner-point abstraction of \cite{BouBri08a}. Following
\cite[Prop.~5]{BouBri08a} (see Appendix~\ref{app:FOG} for
a complete proof), it is possible to find a play $\play$ in $\game$
close to $\overline\play$, in the sense that we control the difference
between their respective weights:
\begin{lemma}\label{lem:fog-exec}
  For all $\varepsilon>0$ and edge
  $\regv\xrightarrow{\overline{\play}}\regv'$ of $\FOG(\rpath)$, there
  exists a play \play in \game following \rpath such that
  $|\weight(\play)-\weight(\overline{\play})|\leq \varepsilon$.
\end{lemma}

In order to prove the direct implication of
Proposition~\ref{prop:timed-scc-sign}, suppose now that \game is
divergent, and consider two simple cycles \rpath and $\rpath'$ in the
same SCC of $\rgame$. We need to show that they have the same sign.
Lemma~\ref{lem:timed-touching-cycles} will first take care of the case
where $\rpath$ and $\rpath'$ share a state $(\state,r)$.

\begin{lemma}\label{lem:timed-touching-cycles}
  If \game is divergent and two cycles $\rpath$ and $\rpath'$ of
  \rgame share a state~$(\state,r)$, they are either both positive or
  both negative.
\end{lemma}
\begin{proof}
  Suppose by contradiction that $\rpath$ is negative and $\rpath'$ is
  positive. We assume that $(\state,r)$ is the first state of both
  $\rpath$ and $\rpath'$, possibly performing cyclic permutations of
  states if necessary. We construct a graph $\FOG(\rpath,\rpath')$ as
  the union of $\FOG(\rpath)$ and $\FOG(\rpath')$ (that share the same
  set of vertices), colouring in blue the edges of $\FOG(\rpath)$ and
  in red the edges of $\FOG(\rpath')$. A path in
  $\FOG(\rpath,\rpath')$ is said blue (respectively, red) when all of
  its edges are blue (respectively, red).
  
  \medskip We assume first that there exists in $\FOG(\rpath,\rpath')$
  a blue cycle $C$ and a red cycle~$C'$ with the same first vertex
  \regv. Let $k$ and $k'$ be the respective lengths of $C$ and $C'$,
  so that $C$ can be decomposed as
  $\regv\xrightarrow{\overline{\play_1}}\cdots
  \xrightarrow{\overline{\play_k}}\regv$ and $C'$ as
  $\regv\xrightarrow{\overline{\play_1'}}\cdots
  \xrightarrow{\overline{\play_{k'}'}}\regv$, where
  $\overline{\play_i}$ are plays following $\overline{\rpath}$ and
  $\overline{\play_i'}$ are plays following $\overline{\rpath'}$, all
  jumping only on corners of regions. Let $\overline{\play}$ be the
  concatenation of $\overline{\play_1},\ldots,\overline{\play_k}$, and
  $\overline{\play'}$ be the concatenation of
  $\overline{\play_1'},\ldots,\overline{\play_{k'}'}$. Recall that
  $w=|\weight(\overline{\play})|$ and
  $w'=|\weight(\overline{\play'})|$ are integers. Since $\rpath$ is
  negative, so is $\rpath^k$, the concatenation of $k$ copies of
  $\rpath$ (the weight of a play following it is a sum of weights all
  below $-1$). Therefore, $\overline{\play}$, that follows $\rpath^k$,
  has a weight $\weight(\overline{\play})\leq -1$. Similarly,
  $\weight(\overline{\play'})\geq 1$. We consider the cycle $C''$
  obtained by concatenating~$w'$ copies of $C$ and $w$ copies of
  $C'$. Similarly, we let $\overline{\play''}$ be the play obtained by
  concatenating $w'$ copies of $\overline{\play}$ and $w$ copies of
  $\overline{\play'}$. By Lemma~\ref{lem:fog-exec}, there exists a
  play $\play''$ in \game, following $C''$ such that
  $|\weight(\play'') - \weight(\overline{\play''})|\leq 1/3$.  But
  $\weight(\overline{\play''})=\weight(\overline{\play})w'+
  \weight(\overline{\play'})w=0$, so $\weight(\play'')\in(-1,1)$: this
  contradicts the divergence of $\game$, since $\play''$ follows the
  cycle of $\rgame$ composed of $w'$ copies $\rpath^k$ and $w$ copies
  of ${\rpath'}^{k'}$ of \rgame.

  \medskip We now return to the general case, where $C$ and $C'$ may
  not exist. Since $\FOG(\rpath)$ and $\FOG(\rpath')$ are finite
  graphs with no deadlocks (every corner has an outgoing edge), from
  every corner of $\FOG(\rpath,\rpath')$, we can reach a blue simple
  cycle, as well as a red simple cycle. Since there are only a finite
  number of simple cycles in $\FOG(\rpath,\rpath')$, there exists a
  blue cycle $C$ and a red cycle $C'$ that can reach each other in
  $\FOG(\rpath,\rpath')$.\footnote{Indeed, one can apply the following
    construction. We start from a fixed vertex and reach a red simple
    cycle. We fix a vertex of this red cycle, and from it we can reach
    a blue simple cycle. We fix a vertex of this blue cycle, and from
    it we can reach a red simple cycle. There is a finite number of
    red and blue simple cycles, so we keep alternating between red and
    blue until we reach a previously seen red simple cycle. This red
    cycle and, for example, the previous blue one can reach each
    other.} In $\FOG(\rpath,\rpath')$, we let $P$ be a path from the
  first vertex of $C$ to the first vertex of $C'$, and $P'$ be a path
  from the first vertex of $C'$ to the first vertex of $C$. Consider
  the cycle $C''$ obtained by concatenating $P$ and $P'$. As a cycle
  of $\FOG(\rpath,\rpath')$, we can map it to a cycle $\rpath''$ of
  \rgame (alternating \rpath and $\rpath'$ depending on the colours of
  the traversed edges), so that $C''$ is a cycle (of length 1) of
  $\FOG(\rpath'')$. By the divergence of $\game$, $\rpath''$ is
  positive or negative. Suppose for instance that it is
  positive. Since $(s,r)$ is the first state of both $\rpath$ and
  $\rpath''$, we can construct the $\FOG(\rpath,\rpath'')$, in which
  $C$ is a blue cycle and $C''$ is a red cycle, both sharing the same
  first vertex. We then conclude with the previous case. A similar
  reasoning with $\rpath'$ applies to the case that $\rpath''$ is
  negative. Therefore, in all cases, we reached a contradiction.\qed
\end{proof}

To finish the proof of the direct implication of
Proposition~\ref{prop:timed-scc-sign}, we suppose that the two simple
cycles $\rpath$ and $\rpath'$ in the same SCC of $\rgame$ do not share
any states. By strong connectivity, in $\rgame$, there exists a path
$\rpath_1$ from the first state of~$\rpath$ to the first state of
$\rpath'$, and a path $\rpath_2$ from the first state of $\rpath'$ to
the first state of $\rpath$. Consider the cycle of $\rgame$ obtained
by concatenating $\rpath_1$ and~$\rpath_2$. By divergence of $\game$,
it must be positive or negative. Since it shares a state with both
$\rpath$ and $\rpath'$, Lemma~\ref{lem:timed-touching-cycles} allows
us to prove a contradiction in both cases. This concludes the proof of
Proposition~\ref{prop:timed-scc-sign}.

\paragraph{Value computation}
We will now explain how to compute the values of a divergent weighted
timed game \game.
Remember that the function $\Val$ maps configurations of
$\States\times\Rplus^\Clocks$ to a value in
$\Rbar=\R\cup\{-\infty,+\infty\}$.  The semi-algorithm of
\cite{BCFL04} relies on the same principle as the value iteration
algorithm used in the untimed setting, only this time we compute the
greatest fixed point of operator
$\IteOpe\colon\Rbar^{\States\times\Rplus^\Clocks} \to
\Rbar^{\States\times\Rplus^\Clocks}$, defined by
$\IteOpe(\vec x)_{(\state,\val)}=0$ if $\state\in\StatesT$, and
otherwise
\[\IteOpe(\vec x)_{(\state,\val)}=
\begin{cases}
  \displaystyle{\sup_{(\state,\val)\xrightarrow{d,\trans}(\state',\val')} 
   d\times\weight(\state)+\weight(\trans)+\vec x_{(\state',\val')}
  } &
  \text{if }\state\in\StatesMax\\
  \displaystyle{\inf_{(\state,\val)\xrightarrow{d,\trans}(\state',\val')}
   d\times\weight(\state)+\weight(\trans)+\vec x_{(\state',\val')} } &
  \text{if }\state\in\StatesMin
\end{cases}\]
\noindent where $(\state,\val)\xrightarrow{d,\trans}(\state',\val')$
ranges over the edges of the infinite weighted game associated with
$\game$ (the one defining its semantics). Then, starting from
$\vec x^0$ mapping every configuration to $+\infty$, we let
$\vec x^i= \IteOpe(\vec x^{i-1})$ for all $i>0$. Since $\vec x^0$ is
piecewise affine (even constant), and $\IteOpe$ preserves 
piecewise affinity, all iterates $\vec x^i$ are piecewise
affine with a finite amount of pieces. In \cite{ABM04}, it is proved
that $\vec x^i$ has at most a number of pieces linear in the size of
\rgame and exponential
in~$i$.\footnote{
  For divergent games with only non-negative weights, the fixed point
  is reached after a number of steps linear in the size of the region
  automaton \cite{BCFL04}: overall, this leads to a doubly
  exponential complexity.}

First, we can compute the set of configurations having value
$+\infty$. Indeed, the region automaton \rgame can be seen as a
reachability two-player game \hgame by saying that $(\state,r)$
belongs to \MinPl (\MaxPl, respectively) if $\state\in\StatesMin$
($\state\in\StatesMax$, respectively). Notice that if
$\Val(\state,\val)=+\infty$, then for all
$\val'\in[\val], \Val(\state,\val')=+\infty$. Therefore, a
configuration $(\state,\val)$ cannot reach the target states if and
only if $(\state,[\val])$ is not in the attractor of \MinPl to the
targets in \hgame. As a consequence, we can compute all such states of
\hgame with complexity linear in the size of~\rgame.

We then decompose \rgame in SCCs. By
Proposition~\ref{prop:timed-scc-sign}, each SCC is either positive or
negative (i.e.\ it contains only positive cycles, or only negative
ones). Then, in order to find the sign of a component, it suffices to
find one of its simple cycles, for example with a depth-first search,
then compute the weight of one play following it.

As we did for weighted (untimed) games, we then compute values in
inverse topological order over the SCCs. Once the values of all
configurations in $(s,r)$ appearing in previously considered SCCs have
been computed, they are no longer modified in further
computation. This is the case, in particular, for all pairs $(s,r)$
that have value $+\infty$, that we precompute from the beginning. In
order to resolve a positive SCC of \rgame, we apply \IteOpe on the
current piecewise affine function, only modifying the pieces appearing
in the SCC, until reaching a fixed point over these pieces. In order
to resolve a negative SCC of \rgame, we compute the attractor for
\MaxPl to the previously computed SCCs: outside of this attractor, we
set the value to $-\infty$. Then, we apply \IteOpe for pieces
appearing in the SCC, initialising them to $-\infty$ (equivalently, we
compute in the dual game, that is a positive SCC), until reaching a
fixed point over these pieces. The next proposition contains the
correction and termination arguments that where presented in
Propositions~\ref{prop:positive-scc}, \ref{prop:minus-infty}, and
\ref{prop:negative-scc} for the untimed setting:
\begin{proposition}\label{prop:VI-timed}
  Let $\game$ be a divergent game with no configurations of value
  $+\infty$.
  \begin{enumerate}
  \item\label{item:positive} The value iteration algorithm applied on
    a positive SCC of $\rgame$ with $n$ states stabilises after at
    most $n$ steps.
  \item\label{item:minus-infinity} In a negative SCC, states $(s,r)$
    of $\rgame$ of value $-\infty$ are all the ones not in the
    attractor for \MaxPl to the targets.
  \item\label{item:negative} The value iteration algorithm,
    initialised with $-\infty$, applied on a negative SCC of $\rgame$
    with $n$ states, and no vertices of value $-\infty$, stabilises
    after at most $n$ steps.
  \end{enumerate}
\end{proposition}

By the complexity results of \cite[Thm.~3]{ABM04}, we obtain a doubly
exponential time algorithm computing the value of a divergent weighted
timed game. This shows that the value problem is in $2$-\EXP\ for
divergent weighted timed game. The proof for \EXP-hardness comes from
a reduction of the problem of solving timed games with reachability
objectives \cite{JurTri07}. To a reachability timed game, we simply
add weights 1 on every transition and 0 on every state, making it a
divergent weighted timed game. Then, \MinPl wins the reachability
timed game if and only if the value in the weighted timed game is
lower than threshold $\alpha=|\States|\times |\regions \Clocks M|$. A
complete proof can be found in
Appendix~\ref{app:exp-hardness}.

In an SCC of $\rgame$, the value iteration algorithm of \cite{ABM04}
allows us to compute an $\varepsilon$-optimal strategy for both
players (for configurations having a finite value), that is constant
(delay or fire a transition) over each piece of the piecewise affine
value function. As in the untimed setting, we may then compose such
$\varepsilon$-optimal strategies to obtain an $\varepsilon'$-optimal
strategy in $\game$ ($\varepsilon'$ is greater than $\varepsilon$, but
can be controlled with respect to the number of SCCs in $\rgame$).

\paragraph{Class decision}
Deciding if a weighted timed game is divergent is \PSPACE-complete.
The complete proof is given in Appendix~\ref{app:class-decision}, and
is an extension of the untimed setting \NL-complete result, but this
time we reason on regions, hence the exponential blowup in complexity:
it heavily relies on Proposition~\ref{prop:timed-scc-sign}, as well as
the corner-point abstraction to keep a compact representation of
plays.

\section{Conclusion}

In this article, we introduced the first decidable class of weighted
timed games with arbitrary weights, with no restrictions on the number
of clocks. Future work include the approximation problem for a larger
class of weighted timed games (divergent ones where we also allow
cycles of weight exactly 0), already studied with only non-negative
weights by~\cite{BJM15}.

\bibliographystyle{plain}

\newpage

\appendix

\section{Technical lemmas regarding value iteration
  algorithms}\label{app:technical}

\begin{lemma}[untimed setting notations]\label{lem:valite}
  For all $i<j\in\N$, if $\vec x^j\neq \vec x^i$ then for all
  $\vertex\in\Vertices$ there exists $\vertex'$ and a play \play from
  \vertex to $\vertex'$ with $|\play|=j-i$ and
  $\weight(\play)=\vec x^j_\vertex-\vec x^i_{\vertex'})$.
\end{lemma}
\begin{proof}
Let us fix $i$, and prove it by induction on $j>i$.

\textit{Initialisation} : If $j=i+1$, we applied one step of the value
iteration algorithm between $\vec x^i$ and $\vec x^j$, so for all
$\vertex$ there exists $\vertex'$ and an edge
$\edge=\vertex\rightarrow\vertex'$ such that $\weight(\edge)=\vec
x^{i+1}_\vertex-\vec x^i_{\vertex'}$.

\textit{Iteration} : We assume the property holds for $j-1>i$, and
$\vec x^j\neq \vec x^i$.  We applied one step of the value iteration
algorithm between $\vec x^{j-1}$ and $\vec x^j$, so for all \vertex
there exists $\vertex'$ and an edge $\edge=\vertex\rightarrow\vertex'$
such that $\weight(\edge)=\vec x^j_\vertex-\vec x^{j-1}_{\vertex'}$.
We apply the property on $i$ and $j-1$ ($\vec x^{j-1}\neq \vec x^i$
because $\vec x^j\neq \vec x^i$ and as soon as $\vec x$ stabilises,
the fixed point is reached and the iteration stops), and obtain that for
all $\vertex'\in\Vertices$ there exists $\vertex''$ and a play \play
from $\vertex'$ to $\vertex''$ with $|\play|=j-1-i$ and
$\weight(\play)=\vec x^{j-1}_{\vertex'}-\vec x^i_{\vertex''}$.  Then
we define
$\play'=\vertex\rightarrow\vertex'\xrightarrow{\play}\vertex''$ and it
holds that $|\play'|=j-i$ and
$\weight(\play')=\vec x^j_{\vertex}-\vec x^i_{\vertex''}$.\qed
\end{proof}

\begin{lemma}[timed setting notations]\label{lem:valite-timed}
  For all $i<j\in\N$, if $\vec x^j\neq \vec x^i$ then, for all
  configurations $(\state,\val)$, there exists $(\state',\val')$ and a
  play \play from $(\state,\val)$ to $(\state',\val')$ with
  $|\play|=j-i$ and
  $\weight(\play)=\vec x^j_{(\state,\val)}-\vec
  x^i_{(\state',\val')}$.
\end{lemma}
\begin{proof}
Let us fix $i$, and prove it by induction on $j>i$.

\textit{Initialisation} : If $j=i+1$, we applied \IteOpe once between
$\vec x^i$ and $\vec x^j$, so for all configurations $(\state,\val)$
there exists $(\state',\val')$ and a transition
$(\state,\val)\xrightarrow{d,\trans}(\state',\val')$ of weight
$\vec x^j_{(\state,\val)}-\vec x^i_{(\state',\val')}$.

\textit{Iteration} : We assume the property holds for $j-1>i$, and
$\vec x^j\neq \vec x^i$.  We applied \IteOpe once between
$\vec x^{j-1}$ and $\vec x^j$, so for all configurations
$(\state,\val)$, there exists $(\state',\val')$ and a transition
$(\state,\val)\xrightarrow{d,\trans}(\state',\val')$ of weight
$\vec x^j_{(\state,\val)}-\vec x^i_{(\state',\val')}$.  We apply the
property on $i$ and $j-1$ ($\vec x^{j-1}\neq \vec x^i$ because
$\vec x^j\neq \vec x^i$ and as soon as $\vec x$ stabilises, the fixed
point is reached and the iteration stops), and obtain that for all
configurations $(\state',\val')$, there exists $(\state'',\val'')$ and
a play \play from $(\state',\val')$ to $(\state'',\val'')$ with
$|\play|=j-1-i$ and
$\weight(\play)=\vec x^{j-1}_{(\state',\val')}-\vec
x^i_{(\state'',\val'')})$.  Then we define
$\play'=(\state,\val) \xrightarrow{d,\trans}(\state',\val')
\xrightarrow{\play}(\state'',\val'')$ and it holds that $|\play'|=j-i$
and
$\weight(\play')=\vec x^j_{(\state,\val)}-\vec
x^i_{(\state'',\val'')}$.\qed
\end{proof}

\section{FOG and corner-point abstraction}\label{app:FOG}

If \val is a valuation and $\varepsilon>0$, $\Ball(\val,\varepsilon)$
denotes the open ball of radius $\varepsilon$ centered into $\val$ for
the infinity norm $\|.\|_\infty$ over $\Rplus^\Clocks$:
$\|\val-\val'\|_\infty = \max_{x\in\Clocks}|\val(x)-\val'(x)|$. We let
$W$ be the biggest weight appearing in \game in absolute value. If
\play and $\overline{\play}$ are plays in \game following respectively
a play \rpath and its copy in $\overline{\game}$, we denote by
$d(\play,\overline{\play})$ the distance between those two plays,
defined as the sum of the differences in absolute value between the
delays on the edges of \play and $\overline{\play}$. By triangular
inequality, we obtain
$|\weight(\play)-\weight(\overline{\play})|\leq W
d(\play,\overline{\play})$, since the same transitions are fired in
\play and $\overline{\play}$, with only different delays. We will now
relate an edge $\regv\xrightarrow{\overline{\play}}\regv'$ of a FOG
with an actual play $\play$ in the original timed game, while
controlling the distance between $\weight(\play)$ and
$\weight(\overline{\play})$:

\begin{lemma}
  Let \rpath be a cycle of \rgame, and
  $\regv\xrightarrow{\overline{\play}}\regv'$ be an edge of
  $\FOG(\rpath)$. For all $\varepsilon>0$ and
  $\val\in r\cap\Ball(\regv,\varepsilon)$, there exists
  $\val'\in r\cap\Ball(\regv',\varepsilon)$ and \play play in \game
  from \val to $\val'$ following \rpath such that
  $|\weight(\play)-\weight(\overline{\play})|\leq 2\varepsilon
  |\rpath| W$.
\end{lemma}
\begin{proof}
  By the previous explanation, it is sufficient to find a play $\play$
  such that $d(\play,\overline{\play})\leq 2\varepsilon |\rpath|$. By
  induction, it is sufficient to prove a similar result only for a
  single edge
  $(\state,r)\xrightarrow{\trans=(\state,g,Y,\state')} (\state',r')$
  of the region automaton $\rgame$, between regions~$r$ and $r'$. We
  thus consider a play
  $(\state,\regv) \xrightarrow{d,\trans}(\state',\regv')$ in the
  closed timed game $\overline{\game}$ from a
  corner~$\regv\in\overline{r}$ to a corner~$\regv'\in\overline{r'}$.
  Consider a valuation $\val\in r\cap\Ball(\regv,\varepsilon)$. We now
  explain how to construct a valuation
  $\val'\in r'\cap\Ball(\regv',\varepsilon)$ and $d'\geq 0$ such that
  $\val\xrightarrow{d',\trans}\val'$ is a valid play in $\game$ and
  $|d-d'|\leq 2\varepsilon$, which implies the lemma.

  Let $r''$ be the time successor region of~$r$ such that
  $r''[Y\leftarrow 0]=r'$. We let $\regv''=\regv+d$ be the valuation
  of~$\overline{r''}$, just before the possible resets of the clocks
  of~$Y$ in~$\overline\game$: $\regv'=\regv''[Y\leftarrow 0]$. Then,
  the timed successors of $\val$, i.e.\ the affine line
  $\val+(1,1,\ldots,1)\R$, intersect the set
  $r''\cap\Ball(\regv'',\varepsilon)$ in a valuation $\val''$: indeed,
  lines obtained by time elapsing starting from \val and \regv are
  parallel, and $r''$ is a time successor of~$r$. There exists $d'$ such that
  $\val''=\val+d'$.  Moreover,
  $d'=\|\val-\val''\|_\infty\leq \|\val-\regv\|_\infty +
  \|\regv-\regv''\|_\infty + \|\regv''-\val''\|_\infty \leq 2
  \varepsilon + d$, and
  $d=\|\regv-\regv'\|_\infty\leq \|\regv-\val\|_\infty +
  \|\val-\val'\|_\infty + \|\val'-\regv'\|_\infty \leq 2 \varepsilon +
  d'$ so that $|d-d'|\leq 2\varepsilon$. Letting
  $\val'=\val''[Y\leftarrow 0]$, we have
  $\val\xrightarrow{d',\trans}\val'$ and
  $\val'\in r'\cap\Ball(\regv',\varepsilon)$. \qed
\end{proof}

This lemma is stronger than Lemma~\ref{lem:fog-exec}, and thus proves
it.

\section{Proofs of correctness and termination of the algorithm for
  divergent weighted timed games}

\begin{proof}[of Proposition~\ref{prop:VI-timed}-\ref{item:positive}]
  Let $W=\max_{a\in\Trans\cup\States} |\weight(a)|$ be the greatest
  weight in the game. There are no negative cycles in the SCC,
  therefore there are no configurations with value $-\infty$, and all
  values are finite. Let $K$ be a bound on the values
  $|\vec x^n_{(\state,\val)}|$ obtained after $n$ steps of the
  algorithm.\footnote{The value iteration emulates the attractor
    computation, so every value is finite after $n$ steps. Moreover,
    functions $(\state,\val)\mapsto \vec x^n_{(\state,\val)}$ are
    piecewise affine with a finite number of pieces over a compact
    space, allowing us to obtain this uniform bound $K$.}  Let us fix
  an integer $p>(2K+(n-1)^2(W+MW))n$.  We will show that the values
  obtained after $n+p$ steps are identical to those obtained after $n$
  steps only.  Therefore, since the algorithm computes non-increasing
  sequences of values, we have indeed stabilised after $n$ steps only.
  Let us assume the existence of a configuration $(\state,\val)$ such
  that $\vec x^{n+p}_{(\state,\val)}<\vec x^{n}_{(\state,\val)}$.  By
  induction on $p$, we can show (see Lemma~\ref{lem:valite-timed} in
  Appendix~\ref{app:technical} for a detailed proof) the existence of
  a configuration $(\state',\val')$ and a finite play \play from
  $(\state,\val)$ to $(\state',\val')$, with length $p$ and weight
  $\vec x^{n+p}_{(\state,\val)}-\vec x^{n}_{(\state',\val')}$: the
  play is composed of the delays and transitions that optimise
  successively the min/max operator in \IteOpe.  This finite play
  being of length greater than $(2K+(n-1)^2(W+MW))n$, if we associate
  each visited configuration $(\state,\val)$ to the state
  $(\state,[\val])$ of \rgame, there is at least one state of \rgame
  appearing more than $2K+(n-1)^2(W+MW)$ times. Thus, it can be
  decomposed into at least $2K+(n-1)^2(W+MW)$ plays following cycles
  of \rgame and at most $(n-1)$ finite plays $\play'_i$ visiting each
  state of \rgame at most once.
  All cycles of the SCC being positive, the weight of \play is at
  least $(2K+(n-1)^2(W+MW)) - (n-1)^2(W+MW)=2K$, bounding from below
  each $\play'_i$'s weight by $-(n-1)(W+MW)$.  Then,
  $\vec x^{n+p}_{(\state,\val)}-\vec x^{n}_{(\state',\val')} \geq 2K$,
  so
  $\vec x^{n+p}_{(\state,\val)} \geq 2K + \vec x^{n}_{(\state',\val')}
  \geq 2K - K \geq K$. But $K \geq \vec x^{n}_{(\state,\val)}$, so
  $\vec x^{n+p}_{(\state,\val)}\geq\vec x^{n}_{(\state,\val)}$, and
  that is a contradiction.\qed
\end{proof}

Much like in the untimed setting, negative SCCs can be resolved using
a dual method. First, we characterise the $-\infty$ values as regions
of \hgame where $\MaxPl$ can not unilaterally guarantee to reach the
targets.

\begin{proof}[of Proposition~\ref{prop:VI-timed}-\ref{item:minus-infinity}]
  Consider a state $(\state,r)$ of \hgame in the attractor for
  $\MaxPl$ to the targets. Then, if \MaxPl applies a winning
  memoryless strategy for the reachability objective to the target
  states, for all $\val\in r$, all strategies of \MinPl will generate
  a play from $(\state,\val)$ reaching a target after at most
  $|\hgame|$ steps. This implies that $(\state,\val)$ has a finite
  (lower) value in the game.

  Reciprocally, if $(\state,r)$ is not in the attractor, by
  determinacy of timed games with reachability objectives, for all
  $\val\in r$, \MinPl has a (memoryless) strategy $\minstrategy$ to
  ensure that no strategy of \MaxPl permits to reach a target state
  from $(\state,\val)$. Applying $\minstrategy$ long enough to
  generate a play following many negative cycles, before switching to
  a strategy allowing \MinPl to reach the target (such a strategy
  exists since no configuration has value $+\infty$ in the game),
  allows \MinPl to obtain from $(\state,\val)$ a negative weight as
  small as possible. Thus, $(\state,\val)$ has value~$-\infty$. \qed
\end{proof}

Thus, given a negative SCC, we can compute configurations of value
$-\infty$ in time polynomial in the SCC's size. Then, finite values
of other configurations can be computed by applying \IteOpe.

\begin{proof}[of Proposition~\ref{prop:VI-timed}-\ref{item:negative}]
  From a negative SCC $\hgame$ that has no more configuration of value
  $+\infty$ or $-\infty$, consider the dual (positive) SCC
  $\widetilde \game$ obtained by:
  \begin{inparaenum}[($i$)]
  \item switching states of $\MinPl$ and $\MaxPl$; 
  \item taking the opposite of every weight in states and transitions.
  \end{inparaenum}
  Sets of strategies of both players are exchanged in those two games,
  so that the upper value in $\game$ is equal to the opposite of the
  lower value in $\widetilde \game$, and vice versa. Since weighted
  games are determined, the value of $\game$ is the opposite of the
  value of $\widetilde\game$. Then, the value of $\game$ can be
  deduced from the value of $\widetilde\game$, for which
  Proposition~\ref{prop:VI-timed}-\ref{item:positive} applies.

  It is then immediate that the values computed with this computation
  of the smallest fixed point of $\IteOpe$ are exactly the opposite
  values of the ones computed in the dual positive SCC.\qed
\end{proof}

\section{\EXP-hardness of the value problem in divergent weighted
  timed games}\label{app:exp-hardness}

Let us show \EXP-hardness of the value problem on divergent weighted
timed games by reducing to reachability in a timed game
\cite{JurTri07}.  Consider an instance of timed game reachability with
a timed game $\mathcal A$ with target states and a configuration
$(\state,\val)$. Let us call \MinPl the player trying to enforce
reachability of a target state in $\mathcal A$ and \MaxPl the one
opposing it. We construct from $\mathcal A$ a weighted timed game
\game by considering the same game with all states weights at $0$ and
all transition weights at $1$. Notice that \game is divergent.  Let us
define $\alpha$ as an upper bound on the number of states in the
region automaton of $\mathcal A$, using classical bounds on the number
of regions: $\alpha=|\States|(4(M+1))^{|\Clocks|}|\Clocks|!$ ($M$ is
the upper bound on clock values).  We consider the instance of the
value problem defined with \game, the configuration $(\state,\val)$,
and the bound $\alpha$.  Then $\Val_\game(\state,\val)\leq \alpha$ if
and only if $(\state,\val)$ can reach the target vertices in
$\mathcal A$.

One direction of this statement's proof is direct by definition of
having a value smaller than $+\infty$, and the other comes from the
fact that reachability in $\mathcal A$ implies reachability in the
region game of $\mathcal A$ in less than $\alpha$ transitions, and
that implies reachability in \game with weight below $\alpha$ from the
configuration $(\state,\val)$ to a target state, therefore
$\Val_\game(\state,\val)\leq \alpha$.

Notice that this reduction is polynomial in the size of $\mathcal A$
because the bound $\alpha$ can be encoded and computed in binary.

\section{Deciding divergence for weighted timed
  games}\label{app:class-decision} 

Let us show how to decide if a game is \emph{not divergent}. By
Proposition~\ref{prop:timed-scc-sign}, it suffices to search for an
SCC of the region automaton containing a non-negative simple cycle
(i.e.\ there exists a play following it of weight in $(-1,+\infty)$)
and a non-positive one (i.e.\ there exists a play following it of
weight in $(-\infty,1)$). Since simple cycles have a length bounded by
$\alpha=|\States|\times |\regions\Clocks
M|=|\States|(4(M+1))^{|\Clocks|}|\Clocks|!$, we simply need to test
the existence of two cycles of length at most $\alpha$, in the same
SCC, one being non-negative and the other non-positive: notice that
this condition indeed implies the non-divergence by the same proof as
in Proposition~\ref{prop:timed-scc-sign}. Finally, notice that, by
Lemma~\ref{lem:fog-exec}, given a cycle of the region automaton, we
can decide if there exists a play following it with weight in
$(-\infty,1)$ (respectively, $(-1,+\infty)$), by guessing a play
following the corners in $\overline\game$ (game where strict
inequalities in guards are replaced with non-strict ones) and checking
that its accumulated weight is an integer in $(\infty,0]$
(respectively, $[0,+\infty)$). We test this condition in $\NPSPACE$,
by the same techniques as in the untimed setting: we guess a starting
region, and a starting corner, for both cycles, we check in polynomial
space that the regions are in the same SCC of $\rgame$, and we guess
on-the-fly the two cycles, i.e.\ the sequences of regions with one of
their corners, keeping in memory their accumulated weight, stopping
when we found two cycles of $\rgame$, or if their length becomes
larger than $\alpha$. Note the accumulated weights are integers
bounded (in absolute value) by
$\alpha\times \max_{a\in\Trans\cup\States} |\weight(a)|$, and can thus
be stored in polynomial space. This shows that deciding the divergence
is in $\co\NPSPACE=\NPSPACE=\PSPACE$ (using the theorems of
Immerman-Szelepcs\'enyi~\cite{Imm88,Sze88} and Savitch~\cite{Sav70}).

Let us now show the $\PSPACE$-hardness (indeed the $\co\PSPACE$, which
is identical) by a reduction from the reachability problem in a timed
automaton. As in the untimed setting, we consider a timed automaton
with a starting state and a different target state without outgoing
transitions. We construct from it a weighted timed game by
distributing all states to \MinPl, and equipping all transitions with
weight $1$, and all states with weight $0$. We also add a loop with
weight $-1$ on the target state, and a transition from the target
state to the initial state with weight $0$, both resetting all
transitions and bounded by a delay of 1. Then, the weighted timed game
is not divergent if and only if the target can be reached from the
initial state in the timed automaton.

\end{document}